\documentclass{article}
\usepackage{amsfonts}
\usepackage{amsthm}
\newtheorem{definition}{Definition}
\newtheorem{theorem}{Theorem}

\usepackage{amsmath}

\usepackage{pstricks}
\usepackage{pstricks-add}
\psset{linewidth=0.3pt,dimen=middle}
\psset{xunit=.70cm,yunit=0.70cm}

\newcommand{\G}{\mathcal G}
\renewcommand{\H}{\mathcal H}
\newcommand{\I}{\mathcal I}
\renewcommand{\P}{\mathbb P}
\newcommand{\C}{\mathbb C}
\newcommand{\E}{\mathbb E}

\newcommand{\blob}{\mathord{\bullet}}
\newcommand{\id}{\operatorname{id}}

\newcommand{\bang}{\operatorname{!}}
\newcommand\rel{\mathrel{\ooalign{\hfil$\mapstochar\mkern5mu$\hfil\cr$\longrightarrow$\cr}}}

\begin{document}


\title{String diagrams for game theory}
\author{Jules Hedges}
\maketitle

\begin{abstract}
This paper presents a monoidal category whose morphisms are games (in the sense of game theory, not game semantics) and an associated diagrammatic language. The two basic operations of a monoidal category, namely categorical composition and tensor product, correspond roughly to sequential and simultaneous composition of games. This leads to a compositional theory in which we can reason about properties of games in terms of corresponding properties of the component parts. In particular, we give a definition of Nash equilibrium which is recursive on the causal structure of the game.

The key technical idea in this paper is the use of continuation passing style for reasoning about the future consequences of players' choices, closely based on applications of selection functions in game theory. Additionally, the clean categorical foundation gives many opportunities for generalisation, for example to learning agents.
\end{abstract}


\section{Introduction}

This paper presents a monoidal category whose morphisms are games (in the sense of game theory, not game semantics) and an associated diagrammatic language. The two basic operations of a monoidal category, namely categorical composition and tensor product, correspond roughly to sequential and simultaneous composition of games. This leads to a compositional theory in which we can reason about properties of games in terms of corresponding properties of the component parts. In particular, we give a definition of Nash equilibrium which is recursive on the causal structure of the game.

The context of this paper begins with \cite{pavlovic09}, which for the first time approached game theory using ideas from program semantics. Many ideas used in this paper first appear there, such as the idea that a game should be a `process' mapping observations to choices, and the use of monads (or equivalently premonoidal categories) to model game-theoretic side effects such as probabilistic choice. Nevertheless, the game theory developed in that paper is no more compositional than ordinary game theory.

Giving a compositional theory of games is far from straightforward. The main objects of study in this paper are so-called \emph{pregames}, which can be thought of as games relative to a continuation. The use of continuation passing style in game theory is implicit in the literature on selection functions (see \cite{escardo11} for an introduction), and was recently made explicit by the author in an (unpublished) extended abstract \cite{hedges15a}. Pregames are also closely related to the `partially defined games' of \cite{oliva14}. Broadly speaking, when a player computes which move to play, she does so with knowledge about how that move is going to be used by the environment (consisting of the rules of the game and the other players) to compute an outcome. If we take this intuition seriously and relativise an \emph{entire} game to a continuation, including the equilibrium condition, it becomes possible to develop a compositional theory of games.

The use of string diagrams as an internal language for monoidal categories has been developed in quantum information theory \cite{baez10} and bialgebra \cite{fauser12}, and also applied in linguistics \cite{coecke13}. A string diagram denotes a morphism of a monoidal category, and the connection is made formal by a \emph{coherence theorem} saying that the morphism denoted by a string diagram is invariant under certain topological moves of the diagram. There are many variants of monoidal categories, each with its own associated string language, surveyed in \cite{selinger11}. The language used in this paper is the one for symmetric monoidal categories (section 3.5 of \emph{loc. cit.}), extended with an operation called \emph{teleological unit}, which is unique to game theory and developed in this paper.

Our string diagrams appear to be related to influence diagrams, an existing graphical language applied mostly in decision theory but also in game theory \cite{koller03}. Although the precise relationship is still to be worked out, string diagrams appear to be preferable because their well-understood categorical semantics (symmetric monoidal categories) allows them to be easily generalised beyond ordinary probabilistic choices. 

The notation used in this paper is intentionally reminiscent of the notation used in linear algebra, quantum theory and linguistics. The categories appearing in those areas have additional structure (namely, a compact closed or dagger structure) which gives a `quantum causality' in which information can appear to flow both forwards and backwards in time. Although we do not have this structure, the teleological unit appears very similar to the unit of a compact structure, and in particular satisfies a coherence theorem very similar to naturality  of the unit. The game-theoretic interpretation of this is that we have a limited form of backward-causality due to rational agents reasoning about future values.

\paragraph{Outline}

Section \ref{category-of-pregames} defines pregames, the objects of study of this paper, defines the various categorical operations on them, and introduces the string diagram language. Section \ref{coherence-theorems} proves the coherence theorems necessary for the string language to be well-defined. Section \ref{selection-functions} details the relationship between pregames and selection functions, on which they are based. Finally section \ref{future-directions} gives several directions for applications and theoretical research.

\paragraph{Conventions}

In this paper we work only with pure strategies, to simplify the proofs in section \ref{coherence-theorems}. This amounts to defining our constructions starting from the category of sets, but all of the theory in this paper works for an arbitrary symmetric monoidal category, of which the most obvious to use is the category of stochastic relations \cite{fong12}. Other possibilities, including generalising to premonoidal categories, are discussed in section \ref{future-directions}.

We will equate sets up to natural isomorphism, which simplifies our notation in several places. The unique element of the terminal set $1$ will be denoted $\bullet$, in order to avoid overuse of the symbol $*$ which denotes duality. We also have deleting and copying functions $\bang : X \to 1$, $\Delta : X \to X \times X$.

A relation $R \subseteq A \times B$ will be denoted $R : A \rel B$, as is usual in the string diagrams literature (for example in \cite{pavlovic09}).

\section{The category of pregames}\label{category-of-pregames}

\begin{definition}
Let $X$, $Y$, $R$ and $S$ be sets. A \emph{pregame} $\G : X \otimes S^* \to Y \otimes R^*$ consists of the following data:
\begin{itemize}
\item A set $\Sigma_\G$ of \emph{strategy profiles}
\item A \emph{play} function $\P_\G : \Sigma_\G \to Y^X$
\item A \emph{coplay} function $\C_\G : \Sigma_\G \to S^{X \times R}$
\item An \emph{individual rationality relation}\footnote{The term `individually rational' is used in economics, in particular in mechanism design, where it is an example of a rationality condition that is both `individual' or `local' (applies to each player individually rather than a group), and is qualitative rather than quantitative.} $\E_\G : \Sigma_\G \rel X \times R^Y$
\end{itemize}
\end{definition}

Formally we will have a category whose objects are pairs of sets, and $\G$ will be a morphism $(X, S) \to (Y, R)$. We will suggestively write this as $\G : X \otimes S^* \to Y \otimes R^*$, as though we had a compact category. Since we do not, this is only syntactic shorthand. We extend this notation in obvious ways, for example writing $A \otimes B^* \otimes C^* \otimes D$ for $(A \times D, B \times C)$. We will draw $\G$ using string diagram notation as
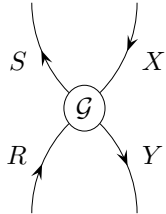
\[\begin{pspicture}(0,0)(15,4)
\psset{ArrowInside=->,arrowscale=2}
\pnode(7,4){Q}
\pnode(9,4){A2}
\rput(8,2){\ovalnode{C}{$\mathcal{G}$}}
\pnode(7,0){P}
\pnode(9,0){C2}
\nccurve[angleA=135,angleB=-90]{C}{Q} \naput{$S$}
\nccurve[angleA=90,angleB=-135]{P}{C} \naput{$R$}
\nccurve[angleA=-90,angleB=45]{A2}{C} \naput{$X$}
\nccurve[angleA=-45, angleB=90]{C}{C2} \naput{$Y$}
\end{pspicture}\]

The most interesting part of the definition of a game is the type $X \times R^Y$. A pair $(x, k) : X \times R^Y$ will be called a \emph{context}, where $x$ is the \emph{history}\footnote{The generalisation beyond the category of sets is not completely obvious here: for a game with side-effects modelled by a monad $M$ on a locally small category the set of contexts should be $\hom (1, M X) \times \hom (Y, M R)$. That is, the `history' is not a concrete history, but an object representing the computation of a history. For example if we have mixed strategies then we need a probability distribution over possible histories.} and $k$ is the \emph{continuation}. The continuation represents `closing the loop' between $Y$ and $R$, with arbitrary future computation done by the environment collapsed into a single function. The composition and tensor product of pregames operate by taking an existing continuation and extending it both forwards and backwards in time, as is usual when programming with delimited continuations.

Also notice the curious similarity between the types of the play/coplay functions and the witnesses used in the Dialectica interpretation of intuitionistic implication \cite{avigad98}. Because of this, there is a strong resemblance between parts of the proof of theorem \ref{theorem-category} and parts of the soundness proof for the Dialectica interpretation.

\begin{definition}
A \emph{closed pregame} is a pregame of the form $\G : S^* \to Y$, with $X = R = 1$. In this case we have $X \times R^Y = 1$, and so $\E_\G$ is a unary relation on strategy profiles $\Sigma_\G$. If $\sigma : \Sigma_\G$ is a strategy profile for a closed pregame $\G$ we will call $\sigma$ an \emph{equilibrium} of $\G$ iff $\sigma \E_\G \blob$.
\end{definition}

In a string diagram, the identity object of a category is denoted by empty space. Therefore a closed pregame is of the form
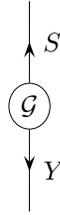
\[ \begin{pspicture}(0,0)(15,4)
\psset{ArrowInside=->,arrowscale=2}
\pnode(8,4){A1}
\rput(8,2){\ovalnode{C1}{$\mathcal{G}$}}
\pnode(8,0){B1}
\nccurve[angleA=90,angleB=-90]{C1}{A1} \nbput{$S$}
\nccurve[angleA=-90,angleB=90]{C1}{B1} \naput{$Y$}
\end{pspicture} \]

We will build pregames, using categorical composition and tensor product, from three atomic components: \emph{decisions}, \emph{computations} and the \emph{teleological unit}\footnote{\emph{Teleology} is the form of causality due to agents `striving' to reach some future aim. The use of this name here is intended to highlight that this operation alone accounts for a large difference between physics and game theory.}. Broadly, the first two divide a game into \emph{players} (or agents) and \emph{rules}. The teleological unit is used to model the backward-causality present in game theory caused by players reasoning about future events.

\begin{definition}
A \emph{decision} is a pregame $\G : X \to Y \otimes R^*$ satisfying $\Sigma_\G = Y^X$ and $\P_\G \sigma = \sigma$.
\end{definition}

Graphically, a decision is of the form
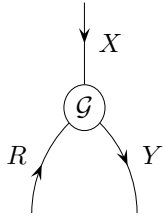
\[ \begin{pspicture}(0,0)(15,4)
\psset{ArrowInside=->,arrowscale=2}
\pnode(8,4){A1}
\rput(8,2){\ovalnode{C1}{$\mathcal{G}$}}
\pnode(7,0){B1}
\pnode(9,0){B2}
\nccurve[angleA=-90,angleB=90]{A1}{C1} \naput{$X$}
\nccurve[angleA=90,angleB=-135]{B1}{C1} \naput{$R$}
\nccurve[angleA=-45, angleB=90]{C1}{B2} \naput{$Y$}
\end{pspicture} \]

To specify a decision is to specify an individual rationality relation $\E_\G : Y^X \rel X \times R^Y$. As for selection functions, this relation can encode a large amount of information about a player's preferences and personality. The intuition is that the decision $\G$ represents an agent who observes the value at $X$ and makes a choice at $Y$, while reasoning forwards in time about the outcome at $R$. The relation $\sigma \E_\G (x, k)$ should hold iff the agent has no incentive to unilaterally deviate from the strategy $\sigma$, in the context in which the history is $x$ and the outcome resulting from the choice $y$ is $ky$. For example, for a classical utility-maximising agent we have $R = \mathbb R$ and
\[ \sigma \E_\G (x, k) \iff k (\sigma x) = \max_{y : Y} k y \]
More generally, in section \ref{selection-functions} we show how to design such relations using multivalued selection functions or quantifiers.

\begin{definition}
Let $f : X \to Y$ be a set-theoretic function. We can view $f$ as a pregame either covariantly as $f : X \to Y$, or contravariantly as $f^* : Y^* \to X^*$. In both cases we have $\Sigma = 1$, and define $\blob \E (x, k)$ to be true for every $x$ and $k$. In the covariant case we set $\P_f \blob = f$, and in the contravariant case we set $\C_{f^*} \blob = f$ (or, adding an explicit type isomorphism for readability, $\C_{f^*} \blob (\blob, x) = f x$).
\end{definition}

Covariant and contravariant computations are respectively drawn as
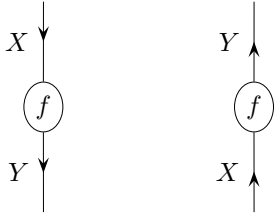
\[ \begin{pspicture}(0,0)(15,4)
\psset{ArrowInside=->,arrowscale=2}
\pnode(5,4){A1}
\rput(5,2){\ovalnode{C1}{$f$}}
\pnode(5,0){B1}
\nccurve[angleA=-90,angleB=90]{A1}{C1} \nbput{$X$}
\nccurve[angleA=-90,angleB=90]{C1}{B1} \nbput{$Y$}
\pnode(9,4){A1}
\rput(9,2){\ovalnode{C1}{$f$}}
\pnode(9,0){B1}
\nccurve[angleA=90,angleB=-90]{C1}{A1} \naput{$Y$}
\nccurve[angleA=90,angleB=-90]{B1}{C1} \naput{$X$}

\end{pspicture}  \]

A particularly important example of a computation is the copying computation $\Delta_X : X \to X \otimes X$. This allows us to use a value more than once, and will be drawn
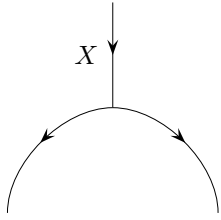
\[ \begin{pspicture}(0,0)(15,4)
\psset{ArrowInside=->,arrowscale=2}
\pnode(7,4){O}
\pnode(7,2){C}
\pnode(5,0){A1}
\pnode(9,0){A2}
\nccurve[angleA=-90,angleB=90]{O}{C} \nbput{$X$}
\nccurve[angleA=180,angleB=90]{C}{A1}
\nccurve[angleA=0,angleB=90]{C}{A2}
\end{pspicture} \]

\begin{definition}
Let $\G : X \otimes T^* \to Y \otimes S^*$ and $\H : Y \otimes S^* \to Z \to R^*$ be pregames. The composition $\H \circ \G : X \otimes T^* \to Z \otimes R^*$ is defined by
\begin{itemize}
\item $\Sigma_{\H \circ \G} = \Sigma_\G \times \Sigma_\H$
\item $\P_{\H \circ \G} (\sigma_1, \sigma_2) = \P_\H \sigma_2 \circ \P_\G \sigma_1$
\item $\C_{\H \circ \G} (\sigma_1, \sigma_2) (x, r) = \C_\G \sigma_1 (x, \C_\H \sigma_2 (\P_\G \sigma_1 x, r))$
\item $(\sigma_1, \sigma_2) \E_{\H \circ \G} (x, k)$ iff $\sigma_1 \E_\G (x, k')$ and $\sigma_2 \E_\H (\P_\G \sigma_1 x, k)$, where
\[ k' y = \C_\H \sigma_2 (y, k (\P_\H \sigma_2 y)) \]
\end{itemize}
\end{definition}

\begin{definition}
Let $\G : X_1 \otimes R_1^* \to Y_1 \otimes S_1^*$ and $\H : X_2 \otimes S_2^* \to Y_2 \otimes R_2^*$ be pregames. The monoidal product
\[ \G \otimes \H : X_1 \otimes S_1^* \otimes X_2 \otimes S_2^* \to Y_1 \otimes R_1^* \otimes Y_2 \otimes R_2^* \]
is defined by
\begin{itemize}
\item $\Sigma_{\G \otimes \H} = \Sigma_\G \times \Sigma_\H$
\item $\P_{\G \otimes \H} (\sigma_1, \sigma_2) (x_1, x_2) = (\P_\G \sigma_1 x_1, \P_\H \sigma_2 x_2)$
\item $\C_{\G \otimes \H} (\sigma_1, \sigma_2) ((x_1, x_2), (r_1, r_2)) = (\C_\G \sigma_1 (x_1, r_1), \C_\H \sigma_2 (x_2, r_2))$
\item $(\sigma_1, \sigma_2) \E_{\G \otimes \H} ((x_1, x_2), k)$ iff $\sigma_1 \E_\G (x_1, k_1)$ and $\sigma_2 \E_\H (x_2, k_2)$ where
\[ k_1 y_1 = (\pi_1 \circ k) (y_1, \P_\H \sigma_2 x_2) \]
\[ k_2 y_2 = (\pi_2 \circ k) (\P_\G \sigma_1 x_1, y_2) \]
\end{itemize}
\end{definition}

Composition and tensor product are graphically represented, as usual, by end-to-end and side-by-side juxtaposition.

\begin{definition}
\emph{Teleological unit} is the pregame $\tau_X : X \otimes X^* \to 1$ given by $\Sigma_{\tau_X} = 1$, $\C_{\tau_X} \bullet (x, \bullet) = x$ and with $\bullet \E_{\tau_X} (x, \bullet)$ holding for all $x$.
\end{definition}

We graphically represent the teleological unit by a \emph{cup}
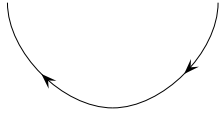
\[ \begin{pspicture}(0,0)(15,2)
\psset{ArrowInside=->,arrowscale=2}
\pnode(5,2){B1}
\pnode(7,0){C1}
\pnode(9,2){C2}
\nccurve[angleA=-90,angleB=0]{C2}{C1}
\nccurve[angleA=180,angleB=-90]{C1}{B1}
\end{pspicture} \]
This notation is usually used for the unit of a compact structure, which comes with a dual \emph{cap} representing the counit, satisfying the intuitive `yanking equation'. Although we do not have a cap, the notation is justified by theorem \ref{teleological-closure-naturality}.

\section{Coherence theorems}\label{coherence-theorems}

\begin{theorem}\label{theorem-category}
There is a category $\mathbf{Pregame}$ whose objects are pairs of sets $(X, R)$, written $X \otimes R^*$, and whose morphisms are pregames. The identity morphism on $X \otimes R^*$ is $\id_{X \otimes R^*} = \id_X \otimes \id_R^*$, and the composition is the one given in the previous section.
\end{theorem}

\begin{proof}
We begin by noting that the identity pregame is explicitly given by
\begin{itemize}
\item $\Sigma_{\id_{X \otimes R^*}} = 1$
\item $\P_{\id_{X \otimes R^*}} \blob = \id_X$
\item $\C_{\id_{X \otimes R^*}} \blob (x, r) = r$
\item $\blob \E_{\id_{X \otimes R^*}} k$ for every $k$
\end{itemize}
(Note that $\id_{X \otimes R^*}$ is really an endomorphism of $(X \times 1, 1 \times R)$, but we are equating sets up to natural isomorphism.)

\paragraph{Left identity}

Let $\G : X \otimes S^* \to Y \otimes R^*$. We prove that $\id_{Y \otimes R^*} \circ \G = \G$. We have
\begin{itemize}
\item $\Sigma_{\id_{Y \otimes R^*} \circ \G} = \Sigma_{\id_{Y \otimes R^*}} \times \Sigma_\G = 1 \times \Sigma_\G = \Sigma_\G$
\item $\P_{\id_{Y \otimes R^*} \circ \G} \sigma = \P_{\id_{Y \otimes R^*}} \blob \circ \P_\G \sigma = \id_Y \circ \P_\G \sigma = \P_\G \sigma$
\item $\C_{\id_{Y \otimes R^*} \circ \G} \sigma (x, r) = \C_\G \sigma (x, \C_{\id_{Y \otimes R^*}} \bullet (\P_\G \sigma x, r)) = \C_\G \sigma (x, r)$
\item $\sigma \E_{\id_{Y \otimes R^*} \circ \G} (x, k) \iff \sigma \E_\G (x, k') \wedge \sigma \E_{\id_{Y \otimes R^*}} (\P_\G \sigma x, k) \iff \sigma \E_\G (x, k')$ where
\[ k' y = \C_{\id_{Y \otimes R^*}} \bullet (y, k (\P_{\id_{Y \otimes R^*}} \bullet y)) = k y \]
\end{itemize}

\paragraph{Right identity}

We prove that $\G \circ \id_{X \otimes S^*} = \G$. We have
\begin{itemize}
\item $\Sigma_{\G \circ \id_{X \otimes S^*}} = \Sigma_\G \times \Sigma_{\id_{X \otimes S^*}} = \Sigma_\G \times 1 = \Sigma_\G$
\item $\P_{\G \circ \id_{Y \otimes R^*}} \sigma = \P_\G \sigma \circ \P_{\id_{X \otimes S^*}} \blob = \P_\G \sigma \circ \id_X = \P_\G \sigma$
\item $\C_{\G \circ \id_{Y \otimes R^*}} \sigma (x, r) = \C_{\id_{X \otimes S^*}} \bullet (x, \C_\G \sigma (\P_{\id_{X \otimes S^*}} \bullet x, r)) = \C_\G \sigma (x, r)$
\item $\sigma \E_{\G \circ \id_{Y \otimes R^*}} (x, k) \iff \bullet \E_{\id_{X \otimes S^*}} (x, k') \wedge \sigma \E_\G (\P_{\id_{X \otimes S^*}} \bullet x, k) \iff \sigma \E_\G (x, k)$
\end{itemize}

\paragraph{Associativity}

Let $\G : X \otimes U^* \to Y \otimes T^*$, $\H : Y \otimes T^* \to Z \otimes S^*$ and $\I : Z \otimes S^* \to W \otimes R^*$. We have
\[ \Sigma_{(\I \circ \H) \circ \G} = \Sigma_\G \times \Sigma_{\I \circ \H} = \Sigma_\G \times \Sigma_\H \times \Sigma_\I = \Sigma_{\H \circ \G} \times \Sigma_\I = \Sigma_{\I \circ (\H \circ \G)} \]
For the play function we have
\begin{align*}
\P_{(\I \circ \H) \circ \G} (\sigma_1, \sigma_2, \sigma_3) &= \P_{\I \circ \H} (\sigma_2, \sigma_3) \circ \P_\G \sigma_1 \\
&= \P_\I \sigma_3 \circ \P_\H \sigma_2 \circ \P_\G \sigma_1 \\
&= \P_\I \sigma_3 \circ \P_{\H \circ \G} (\sigma_1, \sigma_2) \\
&= \P_{\I \circ (\H \circ \G)} (\sigma_1, \sigma_2, \sigma_3)
\end{align*}
and for the coplay function have
\begin{align*}
\C_{(\I \circ \H) \circ \G} (\sigma_1, \sigma_2, \sigma_3) (x, r) &= \C_\G \sigma_1 (x, \C_{\I \circ \H} (\sigma_2, \sigma_3) (\P_\G \sigma_1 x, r)) \\
&= \C_\G \sigma_1 (x, \C_\H \sigma_2 (\P_\G \sigma_1 x, \C_\I \sigma_3 (\P_\H \sigma_2 (\P_\G \sigma_1 x), r))) \\
&= \C_\G \sigma_1 (x, \C_\H \sigma_2 (\P_\G \sigma_1 x, \C_\I \sigma_3 (\P_{\H \circ \G} (\sigma_1, \sigma_2) x, r))) \\
&= \C_{\H \circ \G} (\sigma_1, \sigma_2) (x, \C_\I \sigma_3 (\P_{\H \circ \G} (\sigma_1, \sigma_2) x, r)) \\
&= \C_{\I \circ (\H \circ \G)} (\sigma_1, \sigma_2, \sigma_3) (x, r)
\end{align*}
For the equilibrium condition we have
\begin{align*}
&(\sigma_1, \sigma_2, \sigma_3) \E_{(\I \circ \H) \circ \G} (x, k_3) \\
\iff &\sigma_1 \E_\G (x, k_1) \wedge (\sigma_2, \sigma_3) \E_{\I \circ \H} (\P_\G \sigma_1 x, k_3) \\
\iff &\sigma_1 \E_\G (x, k_1) \wedge \sigma_2 \E_\H (\P_\G \sigma_1 x, k_2) \wedge \sigma_3 \E_\I (\P_\H \sigma_2 (\P_\G \sigma_1 x), k_3) \\
\iff &(\sigma_1, \sigma_2) \E_{\G \circ \H} (x, k_2) \wedge \sigma_3 \E_\I (\P_{\H \circ \G} (\sigma_1, \sigma_2) x, k_3) \\
\iff &(\sigma_1, \sigma_2, \sigma_3) \E_{\I \circ (\H \circ \G)} (x, k_3)
\end{align*}
where
\begin{align*}
k_2 z &= \C_\I (z, k_3 (\P_\I \sigma_3 z)) \\
k_1 y &= \C_{\I \circ \H} (\sigma_2, \sigma_3) (y, k_3 (\P_{\I \circ \H} (\sigma_2, \sigma_3) x)) \\
&= \C_{\I \circ \H} (\sigma_2, \sigma_3) (y, k_3 (\P_\I \sigma_3 (\P_\H \sigma_2 x))) \\
&= \C_\H \sigma_2 (y, \C_\I \sigma_3 (\P_\H \sigma_2 y, k_3 (\P_\I \sigma_3 (\P_\H \sigma_2 x)))) \\
&= \C_\H \sigma_2 (y, k_2 (\P_\H \sigma_2 x))
\end{align*}
\end{proof}

\begin{theorem}
The category $\mathbf{Pregame}$ is symmetric monoidal, with unit $1 \otimes 1^*$, and the tensor product given in the previous section.
\end{theorem}

\begin{proof}
We must prove the existence in $\mathbf{Pregame}$ of natural isomorphisms $\lambda_X : I \otimes X \to X$, $\rho_X : X \otimes I \to X$, $\alpha_{X,Y,Z} : (X \otimes Y) \otimes Z \to X \otimes (Y \otimes Z)$ and $\sigma_{X,Y} : X \otimes Y \to Y \otimes X$ making certain diagrams commute. Since we are treating sets up to natural isomorphism we can take each of these to be the identity, and the diagrams all commute automatically.
\end{proof}

\begin{theorem}\label{teleological-closure-naturality}
If $f : X \to Y$ is a computation then $\tau_Y \circ (f \otimes \id_Y^*) = \tau_X \circ (\id_X \otimes f^*)$. In diagrams,
\[ \begin{pspicture}(0,0)(15,4)
\psset{ArrowInside=->,arrowscale=2}
\pnode(1,4){A1}
\pnode(1,2){B1}
\pnode(3,0){C1}
\rput(5,2){\ovalnode{C2}{$f$}}
\pnode(5,4){A2}
\nccurve[angleA=-90,angleB=90]{A2}{C2}
\nccurve[angleA=-90,angleB=0]{C2}{C1}
\nccurve[angleA=180,angleB=-90]{C1}{B1}
\nccurve[angleA=90,angleB=-90]{B1}{A1}

\rput(7,2){$=$}

\pnode(9,4){A1}
\rput(9,2){\ovalnode{B1}{$f$}}
\pnode(11,0){C1}
\pnode(13,2){C2}
\pnode(13,4){A2}
\nccurve[angleA=-90,angleB=90]{A2}{C2}
\nccurve[angleA=-90,angleB=0]{C2}{C1}
\nccurve[angleA=180,angleB=-90]{C1}{B1}
\nccurve[angleA=90,angleB=-90]{B1}{A1}
\end{pspicture} \]
\end{theorem}

\begin{proof}
Firstly note that we have
\begin{align*}
\P_{f \otimes \id_Y^*} \bullet x &= f x & \C_{f \otimes \id_Y^*} \bullet (x, y) &= y \\
\P_{\id_X \otimes f^*} \bullet x &= x & \C_{\id_X \otimes f^*} \bullet (x, x') &= f x'
\end{align*}
We have $\Sigma_{\tau_Y \circ (f \otimes \id_Y^*)} = 1 = \Sigma_{\tau_X \circ (\id_X \otimes f^*)}$. For the play function,
\[ \P_{\tau_Y \circ (f \otimes \id_Y^*)} \bullet x = \P_{\tau_Y} \bullet (\P_{f \otimes \id_Y^*} \bullet x) = \bullet = \P_{\tau_X} \bullet (\P_{\id_X \otimes f^*} \bullet x) = \P_{\tau_X \circ (\id_X \otimes f^*)} \bullet x \]
and for the coplay function,
\begin{align*}
\C_{\tau_Y \circ (f \otimes \id_Y^*)} \bullet (x, \bullet) &= \C_{f \otimes \id_Y^*} \bullet (x, \C_{\tau_Y} \bullet (\P_{f \otimes \id_Y^*} \bullet x, \bullet)) \\
&= \C_{\tau_Y} \bullet (\P_{f \otimes \id_Y^*} \bullet x, \bullet) \\
&= \P_{f \otimes \id_Y^*} \bullet x \\
&= f x \\
&= f (\P_{\id_X \otimes f^*} \bullet x) \\
&= f (\C_{\tau_X} \bullet (\P_{\id_X \otimes f^*} \bullet x, x)) \\
&= \C_{\id_X \otimes f^*} (x, \C_{\tau_X} \bullet (\P_{\id_X \otimes f^*} \bullet x, \bullet)) \\
&= \C_{\tau_X \circ (\id_X \otimes f^*)} \bullet (x, \bullet)
\end{align*}
Finally for both cases we have $\bullet \E (x, \bullet)$ for all $x$.
\end{proof}

\section{Relationship to selection functions}\label{selection-functions}

In this section we show that pregames subsume certain classes of `higher-order games'. Firstly we consider the `context-dependent games' of \cite{hedges14b}, which provide a large generalisation of simultaneous games with pure strategies. Secondly we consider finite generalised sequential games in the sense of \cite{escardo11}, which provide a large generalisation of extensive-form games of perfect information. In each case we will focus on two-players games for simplicity. We can also easily generalise to mixed strategies, which is described in section \ref{future-directions}.

A two-player context-dependent game is defined in \cite{hedges14b} to consist of the following data:
\begin{itemize}
\item Sets $X$, $Y$ of choices for each player, and $R$ of outcomes
\item Multivalued selection functions $\varepsilon : (X \to R) \to \mathcal P X$, $\delta : (Y \to R) \to \mathcal P Y$
\item An outcome function $q : X \times Y \to R$
\end{itemize}
A strategy profile is simply a pair $(\sigma_1, \sigma_2) : X \times Y$. A strategy profile is called a \emph{selection equilibrium} if
\[ \sigma_1 \in \varepsilon \lambda x . q (x, \sigma_2) \]
\[ \sigma_2 \in \delta \lambda y . q (\sigma_1, y) \]

\begin{theorem}
The selection equilibria of this game are precisely the equilibria of the string diagram
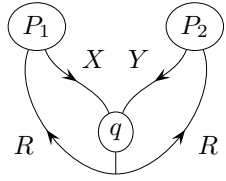
\[ \begin{pspicture}(0,0)(15,4)
\rput(7,1){\ovalnode{U}{$q$}}
\rput(5.5,3){\ovalnode{P1}{$P_1$}}
\rput(8.5,3){\ovalnode{P2}{$P_2$}}
\pnode(7,0.2){Ud}
\nccurve[angleA=-70,angleB=110,ArrowInside=->,arrowscale=2]{P1}{U} \naput{$X$}
\nccurve[angleA=-110,angleB=70,ArrowInside=->,arrowscale=2]{P2}{U} \nbput{$Y$}
\nccurve[angleA=-90,angleB=90]{U}{Ud} 
\nccurve[angleA=180,angleB=-110,ArrowInside=->,arrowscale=2]{Ud}{P1} \naput{$R$}
\nccurve[angleA=0,angleB=-70,ArrowInside=->,arrowscale=2]{Ud}{P2} \nbput{$R$}
\end{pspicture} \]
where
\[ \sigma_1 \E_{P_1} (\bullet, k) \iff \sigma_1 \in \varepsilon k \]
\[ \sigma_2 \E_{P_2} (\bullet, k) \iff \sigma_2 \in \delta k \]
\end{theorem}

\begin{proof}
Algebraically, this string diagram is
\[ \tau_R \circ (q \otimes \Delta_R^*) \circ (P_1 \otimes P_2) \]
Unwinding the definitions, we have
\begin{align*}
&(\sigma_1, \sigma_2) \E_{\tau_R \circ (q \otimes \Delta_R^*) \circ (P_1 \otimes P_2)} (\bullet, \bullet) \\
\iff &(\sigma_1, \sigma_2) \E_{P_1 \otimes P_2} (\bullet, k) \wedge \bullet \E_{\tau_R \circ (q \otimes \Delta_R^*)} (\P_{P_1 \otimes P_2} (\sigma_1, \sigma_2) \bullet, \bullet) \\
\iff &(\sigma_1, \sigma_2) \E_{P_1 \otimes P_2} (\bullet, k) \\
\iff &\sigma_1 \E_{P_1} (\bullet, k_1) \wedge \sigma_2 \E_{P_2} (\bullet, k_2) \\
\iff &\sigma_1 \in \varepsilon k_1 \wedge \sigma_2 \in \delta k_2
\end{align*}
where
\begin{align*}
k (x, y) &= \C_{\tau_R \circ (q \otimes \Delta_R^*)} \bullet ((x, y), \bullet) \\
&= \C_{q \otimes \Delta_R^*} \bullet ((x, y), \C_{\tau_R} \bullet (\P_{q \otimes \Delta_R^*} \bullet (x, y), \bullet)) \\
&= \C_{q \otimes \Delta_R^*} \bullet ((x, y), \C_{\tau_R} \bullet (q (x, y), \bullet)) \\
&= \C_{q \otimes \Delta_R^*} \bullet ((x, y), q (x, y)) \\
&= \Delta_R (q (x, y)) \\
&= (q (x, y), q (x, y)) \\
\end{align*}
\[ k_1 x = (\pi_1 \circ k) (x, \P_{P_2} \sigma_2 \bullet) = (\pi_1 \circ k) (x, \sigma_2) = q (x, \sigma_2) \]
\[ k_2 y = (\pi_2 \circ k) (\P_{P_1} \sigma_1 \bullet, y) = (\pi_2 \circ k) (\sigma_1, y) = q (\sigma_1, y) \]
\end{proof}

A two-player sequential game is defined in \cite{escardo11} to consist of the following data:
\begin{itemize}
\item Sets $X$, $Y$ of choices for each player, and $R$ of outcomes
\item Multivalued quantifiers $\varphi : (X \to R) \to \mathcal P R$, $\psi : (Y \to R) \to \mathcal P R$
\item An outcome function $q : X \times Y \to R$
\end{itemize}
A strategy profile for this game consists of a move $\sigma_1 : X$ for the first player and a contingent strategy $\sigma_2 : X \to Y$ for the second player. A strategy profile is called \emph{optimal} if
\begin{align*}
q (\sigma_1, \sigma_2 \sigma_1) &\in \varphi \lambda x . q (x, \sigma_2 x) \\
q (x, \sigma_2 x) &\in \delta \lambda y . q (x, y) \hbox{ for all } x : X
\end{align*}
(Note that the difference between selection functions and quantifiers is relatively unimportant: we could equally well define simultaneous games using quantifiers, and sequential games using selection functions.)

\begin{theorem}
An optimal strategy profile for this game is an equilibrium of the string diagram
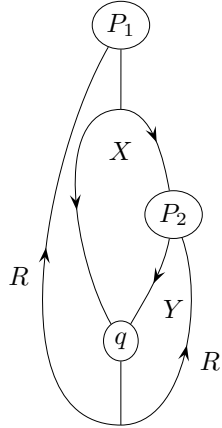
\[ \begin{pspicture}(0,0)(8,9)
\rput(4,8){\ovalnode{P1}{$P_1$}}
\pnode(4,6.4){P1d}
\rput(5,4.4){\ovalnode{P2}{$P_2$}}
\rput(4,2){\ovalnode{q}{$q$}}
\pnode(4,0.4){qd}
\nccurve[angleA=-90,angleB=-90]{P1}{P1d}
\nccurve[angleA=0,angleB=100,ArrowInside=->,arrowscale=2]{P1d}{P2}
\nccurve[angleA=-100,angleB=60,ArrowInside=->,arrowscale=2]{P2}{q}
\nccurve[angleA=180,angleB=120,ArrowInside=->,arrowscale=2]{P1d}{q}
\nccurve[angleA=-90,angleB=-90]{q}{qd}
\nccurve[angleA=0,angleB=-70,ArrowInside=->,arrowscale=2]{qd}{P2} \nbput{$R$}
\nccurve[angleA=180,angleB=-120,ArrowInside=->,arrowscale=2]{qd}{P1} \naput{$R$}
\rput(4,5.6){$X$}
\rput(5,2.6){$Y$}
\end{pspicture} \]
where
\begin{align*}
\sigma_1 \E_{P_1} (\bullet, k_1) &\iff k_1 \sigma_1 \in \varphi k_1 \\
\sigma_2 \E_{P_2} (x, k_2) &\iff k_2 (\sigma_2 x) \in \psi k_2
\end{align*}
\end{theorem}

\begin{proof}
Algebraically this string diagram is
\[ \tau_R \circ (q \otimes \Delta_R^*) \circ (((\id_X \otimes P_2) \circ \Delta_X) \otimes \id_R^*) \circ P_1 \]
Unwinding the definition, we have
\begin{align*}
&(\sigma_1, \sigma_2) \E_{\tau_R \circ (q \otimes \Delta_R^*) \circ (((\id_X \otimes P_2) \circ \Delta_X) \otimes \id_R^*) \circ P_1} (\bullet, \bullet) \\
\iff &(\sigma_1, \sigma_2) \E_{(((\id_X \otimes P_2) \circ \Delta_X) \otimes \id_R^*) \circ P_1} (\bullet, k_1) \\
\iff &\sigma_1 \E_{P_1} (\bullet, k_2) \wedge \sigma_2 \E_{((\id_X \otimes P_2) \circ \Delta_X) \otimes \id_R^*} (\P_{P_1} \sigma_1 \bullet, k_1) \\
\iff &\sigma_1 \E_{P_1} (\bullet, k_2) \wedge \sigma_2 \E_{(\id_X \otimes P_2) \circ \Delta_X} (\sigma_1, k_3) \\
\iff &\sigma_1 \E_{P_1} (\bullet, k_2) \wedge \sigma_2 \E_{\id_X \otimes P_2} (\P_{\Delta_X} \bullet \sigma_1, k_3) \\
\iff &\sigma_1 \E_{P_1} (\bullet, k_2) \wedge \sigma_2 \E_{P_2} (\sigma_1, k_4) \\
\iff &q (\sigma_1, \sigma_2 \sigma_1) \in \varphi \lambda x . q (x, \sigma_2 x) \wedge q (\sigma_1, \sigma_2 \sigma_1) \in \psi \lambda y . q (\sigma_1, y)
\end{align*}
where
\begin{align*}
k_1 (x, y) &= (q (x, y), q (x, y)) \\
k_2 x &= \C_{((\id_X \otimes P_2) \circ \Delta_X) \otimes \id_R^*} \sigma_2 (x, k_1 (\P_{((\id_X \otimes P_2) \circ \Delta_X) \otimes \id_R^*} \sigma_2 x)) \\
&= (\pi_2 \circ k_1) (\P_{((\id_X \otimes P_2) \circ \Delta_X) \otimes \id_R^*} \sigma_2 x) \\
&= (\pi_2 \circ k_1) (\P_{(\id_X \otimes P_2) \circ \Delta_X} \sigma_2 x) \\
&= (\pi_2 \circ k_1) (\P_{\id_X \otimes P_2} \sigma_2 (x, x)) \\
&= (\pi_2 \circ k_1) (\P_{\id_X} \bullet x, \P_{P_2} \sigma_2 x) \\
&= (\pi_2 \circ k_1) (x, \sigma_2 x) \\
&= q (x, \sigma_2 x) \\
k_3 (x, y) &= (\pi_1 \circ k_1) (x, y) \\
&= q (x, y) \\
k_4 y &= k_3 (\P_{\id_X} \bullet \sigma_1, y) \\
&= q (\sigma_1, y)
\end{align*}
\end{proof}

The converse does not hold, because the definition of optimal strategy profile generalises \emph{subgame-perfect equilibria}, which is an equilibrium refinement of Nash \cite{escardo12}. The definitions given in this paper generalise easily to subgame-perfect equilibria, but we do not do this because subgame-perfection, unlike Nash, is not decidable in general.

\section{Future directions}\label{future-directions}

The potential applications of a compositional, graphical game theory are numerous, especially in economics, and this paper also raises some interesting theoretical questions. We conclude by broadly giving some future research directions and questions, most of which are being explored by the author together with the other authors of \cite{hedges14b}.

\begin{itemize}
\item Possibly the most important theoretical concepts missing from this paper are repeated games and incomplete information games, both of which are ubiquitous in economic applications. This is work in progress.

\item Another important aspect of game theory that cannot be modelled in this way is the ability for the `shape' of a subgame to depend on a previous move, for example with the `moves of nature' used in the usual approach to incomplete information. The obvious approach to this is to use dependent types, which leads immediately into current research on type systems. A complication is that mechanical type inference has already proven invaluable in practice.

\item A promising approach to the semantics of the teleological unit is a strong resemblance to shift/reset operators for programming with delimited continuations \cite{danvy90,asai11}: a decision is analogous to the shift operator in that it captures a continuation, and the teleological unit is analogous to the reset operator in that it delimits a continuation.

\item Due to the difficulty of reasoning with continuations, computer support is vital for all but the most trivial applications. The author has developed a Haskell implementation, but it is extremely awkward to use because the Haskell type system does not unify types like $X \sim 1 \times X$ and $X \sim 1 \to X$, and so the user must manually keep track of these isomorphisms. As an intermediate step, a code generator for a domain specific language similar to Haskell's arrows \cite{paterson01} would be useful. (Unfortunately, for technical reasons it does not seem to be possible to use GHC's built-in arrow preprocessor.) Ultimately a graphical interface would be invaluable for these ideas to become accessible to working economists.

\item As a by-product of obtaining a compositional theory, we have the ability to model preferences of agents which are extremely different to utility maximisation or preference relations. This extends a line of work begun in \cite{hedges14b}, which uses fixpoint selection functions to model coordination and differentiation. Obvious next steps include modelling bounded rationality \cite{rubinstein98} and social concerns.

\item A potentially very powerful dimension is to vary the underlying category, as discussed in the introduction. The use of ordinary (possibilistic) nondeterminism in game theory is explored for example in \cite{pavlovic09,blumensath13,hedges14} and \cite[chapter 9]{lavalle06}, and work in progress by the author suggests that the order structure on possibilistic strategies is important. We also have experimental evidence that correlated equilibria \cite{aumann74} appear as a special case by using a commutative monad transformer stack in which a reader monad gives players read-only access to a shared randomising device.

\item Using noncommutative side-effects is potentially even more rewarding, but there is theoretical work to be done on graphical languages for Freyd categories. A major aim is to use strategies with mutable states to model learning, and individual rationality relations to specify that a strategy can be subjectively rational with respect to current epistemic knowledge, for example using methods of epistemic game theory \cite{perea12}.
\end{itemize}

\bibliographystyle{plain}
\bibliography{/Users/jules/Dropbox/Work/refs}

\end{document}